\documentclass[10pt,a4paper]{article}
\usepackage{jucs2e}
\usepackage{amsfonts}
\usepackage{amsmath, amssymb}
\usepackage{mathtools}

\usepackage{float}
\usepackage{booktabs, multirow}
\usepackage{microtype}
\usepackage{mathtools}
\usepackage{epstopdf}
\usepackage[caption=false]{subfig}
\usepackage{comment}
\usepackage{tikz}
\usepackage{color}
\definecolor{amber}{rgb}{0, 0.50, 1.0}
\newcommand{\novo}[1]{#1}

\makeatletter
\newcommand{\@chapapp}{\relax}%
\makeatother

\usetikzlibrary{arrows}
\usetikzlibrary{arrows.meta}
\makeatletter
\newcommand*{\textoverbar}[1]{$\bar{\hbox{#1}}\m@th$}
\renewcommand{\iff}{\leftrightarrow}
\newcommand{\comp}{\cdot}
\newcommand{\signedrho}{\bar{\rho}}

\makeatother

\usepackage{microtype}

\newcommand\blfootnote[1]{%
  \begingroup
  \renewcommand\thefootnote{}\footnote{#1}%
  \addtocounter{footnote}{-1}%
  \endgroup
}

\usepackage{lipsum}
\usepackage{amsfonts}
\usepackage{graphicx}
\usepackage{epstopdf}
\usepackage{algorithmic}
\ifpdf
  \DeclareGraphicsExtensions{.eps,.pdf,.png,.jpg}
\else
  \DeclareGraphicsExtensions{.eps}
\fi

\begin{document}

\title{On the Complexity of Some Variations of Sorting by Transpositions\blfootnote{Journal of Universal Computer Science 26(9), https://doi.org/10.3897/jucs.2020.057}}

\author{
   {\bfseries Alexsandro~Oliveira~Alexandrino}\\
    (Institute of Computing, University of Campinas\\
    Campinas, Brazil \\
    alexsandro@ic.unicamp.br)
    \and
    {\bfseries Andre~Rodrigues~Oliveira}\\
     (Institute of Computing, University of Campinas\\
     Campinas, Brazil \\
     andrero@ic.unicamp.br)
     \and
   {\bfseries Ulisses~Dias}\\
    (School of Technology, University of Campinas \\ Limeira, S{\~a}o Paulo, Brazil \\
    ulisses@ft.unicamp.br)
    \and
   {\bfseries Zanoni~Dias}\\
      (Institute of Computing, University of Campinas\\
      Campinas, Brazil \\
      zanoni@ic.unicamp.br)
   \\
}

\maketitle

\begin{abstract}
  One of the main challenges in Computational Biology is to find the evolutionary distance between two organisms. In the field of comparative genomics, one way to estimate such distance is to find a minimum cost sequence of rearrangements (large scale mutations) needed to transform one genome into another, which is called the rearrangement distance. In the past decades, these problems were studied considering many types of rearrangements (such as reversals, transpositions, transreversals, and revrevs) and considering the same weight for all rearrangements, or different weights depending on the types of rearrangements. The complexity of the problems involving reversals, transpositions, and both rearrangements is known, even though the hardness proof for the problem combining reversals and transpositions was recently given. In this paper, we enhance the knowledge for these problems by proving that models involving transpositions alongside reversals, transreversals, and revrevs are NP-hard, considering weights $w_1$ for reversals and $w_2$ for the other rearrangements such that $w_2/w_1 \leq 1.5$. In addition, we address a cost function related to the number of fragmentations caused by a rearrangement, proving that the problem of finding a minimum cost sorting sequence, considering the fragmentation cost function with some restrictions, is NP-hard for transpositions and the combination of reversals and transpositions.
\end{abstract}


\begin{category}
F.2.0, G.2.1.
\end{category}


\section{Introduction}

A \emph{Genome Rearrangement} is a large scale mutation which changes the position and the orientation of conserved regions in a genome. In the field of comparative genomics, one common approach to estimate the evolutionary distance is to formulate it as the minimum number of rearrangements that transforms a genome into another, which is called \emph{rearrangement distance}.

In comparative genomics, a genome is represented as an ordered sequence of conserved blocks (high similarity regions) and, depending on the genomic information available, different mathematical models can be used. Considering that a genome has no repeated conserved blocks, we can model the genome as a permutation, where each element represents a conserved block. If the orientation of the genes is known, we use signed permutations to indicate the orientation of the elements. If the orientation is unknown, we use unsigned permutations. When using this representation, the problem of finding the rearrangement distance between two genomes is equivalent to the problem of finding the sorting by rearrangements distance of a permutation~\cite{1995-kececioglu-sankoff}, which is the minimum number of operations needed to transform this permutation into a permutation where each element is positive (or unsigned) and in ascending order.

A \emph{rearrangement model} defines the set of rearrangements allowed to compute the distance.
Two of the most studied rearrangements are reversals, which inverts a segment of the genome, and transpositions, which swaps the position of two adjacent segments of the genome. Previous works focused on the problems of Sorting by Reversals, Transpositions, and both rearrangements. \cite{1995-hannenhalli-pevzner} presented a polynomial algorithm for Sorting Signed Permutations by Reversals. For the unsigned case, \cite{1999-caprara} showed that the problem is NP-hard. Since transpositions do not change the sign of the elements, when considering only transpositions we have the problem of Sorting Unsigned Permutations by Transpositions, which was also proved to be NP-hard~\cite{2012-bulteau-etal}. Despite having approximation algorithms proposed since the late 1990's~\cite{1998-walter-etal}, the complexity of the problems of Sorting (Signed or Unsigned) Permutations by Reversals and Transpositions had been unknown until recently, when \cite{2019b-oliveira-etal} presented a proof that these problems are NP-hard. A block-interchange is an operation that swaps any two segments of the genome without changing the orientation of the elements. The problems of Sorting Unsigned Permutations by Block-Interchanges and Sorting Signed Permutations by Reversals and Block-Interchanges are solvable in polynomial time~\cite{1996-christie, 2007-mira-meidanis}.

Other rearrangement operations are transreversals and revrevs. Given two adjacent segments of a genome, a transreversal is an operation that swaps these two segments and inverts the elements of one of these segments, while a revrev is an operation that inverts the elements for each of these two segments. Although the complexity of the problems involving these operations was unknown, many approximation algorithms were presented in the literature~\cite{2009-fertin-etal}. \cite{1999-gu-etal} presented a $2$-approximation algorithm for signed permutations considering reversals, transpositions, and transreversals. \cite{2001-lin-xue} added the revrev operation to the model and gave a $1.75$-approximation algorithm. For the model containing transpositions, transreversals, and revrevs, the best result for signed permutations is a $1.5$-approximation algorithm~\cite{2005-hartman-sharan}. \cite{2010-lou-zhu} presented a $2.25$-approximation algorithm for Sorting Unsigned Permutations by Reversals, Transpositions, and Transreversals.

The traditional approach for the genome rearrangements problems is to consider that every rearrangement has the same cost and, thus, the sorting distance consists in finding a minimum length sorting sequence of rearrangements. The weighted approach was motivated by the observation that some rearrangements are more likely to occur than others~\cite{2007-bader-ohlebusch,1996-blanchette-etal}. In a weighted approach, each rearrangement has an associated cost and the goal is to find a minimum-cost sorting sequence of rearrangements. We use $w_1$ to represent the weight of reversals and $w_2$ to represent the weights of transpositions, transreversals, and revrevs. For values of $w_1$ and $w_2$ such that $1 \leq w_2 / w_1 \leq 2$, \cite{2007-bader-ohlebusch} gave a $1.5$-approximation algorithm for the model containing reversals, transpositions, and transreversals on signed permutations. For the same problem and considering $w_2 / w_1 = 2$, \cite{2002-eriksen} presented a $7/6$-approximation and a polynomial-time approximation scheme. We show that the problems of Sorting (Signed or Unsigned) Permutations by Rearrangements are NP-hard for rearrangement models that include transpositions or the combination of reversals and transpositions alongside transreversals and revrevs, considering that $w_2/w_1 \leq 1.5$.

\cite{2018-alexandrino-etal} introduced a new cost function equal to the number of fragmentations (i.e., breaks of adjacent elements) caused by a rearrangement, and they presented approximation algorithms for models containing reversals and transpositions, considering this cost function. \cite{2020-alexandrino-etal} also considered fragmentation-weighted operations, presenting better approximation algorithms for some permutation classes. In this approach, prefix and suffix operations, which respectively modify the beginning and the end of the genome, cause less fragmentation in the genome and, so, they cost less. Using a parsimony criterion, the problem is modeled so that the number of fragmentations in the genome is minimized during the sorting process. \novo{We show that the problems of Sorting (Signed or Unsigned) Permutations by Transpositions, or by Reversals and Transpositions, are NP-hard considering the fragmentation cost function proposed by \cite{2018-alexandrino-etal, 2020-alexandrino-etal}. We also show that these problems are NP-hard for other combinations of weights related to the number of fragmentations caused by the operations.}

This work is organized as follows. Section~\ref{sec:definition} presents definitions and notations related to the problems. Section~\ref{sec:trans_revrev} shows hardness proofs for the models containing transreversals and revrevs. Section~\ref{sec:frag} shows hardness proofs for the fragmentation-weighted problems containing transpositions.
At last, Section~\ref{sec:conclusion} presents final considerations and future work.

\section{Definitions}\label{sec:definition}

Considering the case where a genome does not have repeated genes, a genome $\mathcal{G}$ is modeled as a permutation whose elements represent conserved blocks. If the orientation of the genes is known, $\mathcal{G}$ is represented as a signed permutation, and $\mathcal{G}$ is represented as an unsigned permutation otherwise. In this case, the problem of finding the rearrangement distance between two genomes $\mathcal{G}_1$ and $\mathcal{G}_2$ is equivalent to the problem of finding the sorting rearrangement distance of a permutation~\cite{2009-fertin-etal}.

A \emph{signed permutation} is represented as $\pi = (\pi_1~\pi_2~\ldots~\pi_n)$, such that $\pi_i \in \{{-n}, \ldots, -1, +1, \dots, +n \}$ and $|\pi_i| \neq |\pi_j|$ $\iff$ $i \neq j$, for all $i$ and $j$. An \emph{unsigned permutation} is also represented as $\pi = (\pi_1~\pi_2~\ldots~\pi_n)$, but $\pi_i \in \{1, 2, \dots, n \}$ and $\pi_i \neq \pi_j$ $\iff$ $i \neq j$, for all $i$ and $j$. The \emph{identity permutation}, $\iota = (1~2~\ldots~n)$, is the sorted permutation, and it is the target of the sorting problems. For signed permutations, we have $\iota = ({+1}~{+2}~\ldots~{+n})$. The \emph{reverse permutation} is defined as $\eta = (n~(n-1)~\ldots~1)$, for unsigned permutations, and $\bar{\eta} = ({-n}~{-(n-1)~\ldots~1})$, for signed permutations.

A \emph{rearrangement model} $\mathcal{M}$ is the set of allowed operations in a rearrangement problem. Considering the unweighted approach, given a rearrangement model $\mathcal{M}$ and a permutation $\pi$, the sorting distance $d_{\mathcal{M}}(\pi)$ is equal to the minimum number of rearrangements from $\mathcal{M}$ that sorts the permutation $\pi$.

Considering a weighted function $w: \mathcal{M} \rightarrow \mathbb{R}$, given a rearrangement $\mathcal{M}$ and a permutation $\pi$, the sorting distance $d_{\mathcal{M}}(\pi)$ is equal to $\sum_{i=1}^{\ell} w(\beta_i)$ such that $\beta_i \in \mathcal{M}$, for $1 \leq i \leq \ell$, $\pi \comp \beta_1 \comp \ldots \comp \beta_\ell = \iota$, and $\sum_{i=1}^{\ell} w(\beta_i)$ is minimum.
For a sequence of rearrangements $S = \beta_1, \beta_2, \ldots, \beta_\ell$, we have that $w(S) = \sum_{i = 1}^{\ell} w(\beta_i)$. The unweighted approach is equivalent to using a unitary weight for all rearrangements.

A \emph{reversal} is a rearrangement which inverts a segment of the genome and, when applied to a signed permutation, flips the sign of the elements in this segment. A \emph{transposition} is a rearrangement that exchanges the position of two adjacent segments of the genome. Next, we formally define these operations.

\begin{definition}
  Considering an unsigned permutation $\pi$, a reversal $\rho(i,j)$, with $1 \leq i < j \leq n$, is an operation that when applied to $\pi$ transforms it in the permutation $\pi \comp \rho(i,j) = (\pi_1~\ldots~\pi_{i-1}~\underline{\pi_j~\pi_{j-1}~\ldots~\pi_{i+1}~\pi_i}~\pi_{j+1}~\ldots~\pi_n)$.
\end{definition}

\begin{definition}
  Considering a signed permutation $\pi$, a reversal $\signedrho(i,j)$, with $1 \leq i \leq j \leq n$, is an operation that when applied to $\pi$ transforms it in the permutation $\pi \comp \signedrho(i,j) = (\pi_1~\ldots~\pi_{i-1}~\underline{{-\pi_j}~{-\pi_{j-1}}~\ldots~{-\pi_{i+1}}~{-\pi_i}}~\pi_{j+1}~\ldots~\pi_n)$.
\end{definition}

\begin{definition}
  A transposition $\tau(i,j,k)$, with $1 \leq i < j < k \leq {n+1}$, is an operation that when applied to $\pi$ transforms it in the permutation $\pi \comp \tau(i,j,k) = (\pi_1~\ldots~\pi_{i-1}~\underline{\pi_j~\ldots~\pi_{k-1}}~\underline{\pi_i~\ldots~\pi_{j-1}}~\pi_{k}~\ldots~\pi_n)$.
\end{definition}

For two adjacent segments $A$ and $B$, a \emph{transreversal} is a rearrangement that inverts the elements of $A$ (type one) or $B$ (type two) and exchanges the position of $A$ and $B$. Also, when applied to a signed permutation, it flips the sign of the elements in the inverted segment. A \emph{revrev} inverts each of two adjacent segments and, when applied to a signed permutation, flips the sign of the elements affected.

\begin{definition}
  Considering an unsigned permutation $\pi$, a transreversal Type 1 ${\rho\tau}_1(i,j,k)$ and a transreversal Type 2 ${\rho\tau}_2(i,j,k)$, with $1 \leq i < j < k \leq {n+1}$, are operations that transform a permutation $\pi$ in the following way:
  \small{\begin{align*}
    \pi \comp {\rho\tau}_1(i,j,k) = (\ldots~\pi_{i-1}~\underline{\pi_j~\ldots~\pi_{k-1}}~\underline{\pi_{j-1}~\ldots~\pi_{i}}~\pi_{k}~\ldots), \\
    \pi \comp {\rho\tau}_2(i,j,k) = (\ldots~\pi_{i-1}~\underline{\pi_{k-1}~\ldots~\pi_{j}}~\underline{\pi_i~\ldots~\pi_{j-1}}~\pi_k~\ldots).
  \end{align*}}
\end{definition}

\begin{definition}
  Considering a signed permutation $\pi$, a transreversal Type 1 $\bar{\rho\tau}_1(i,j,k)$ and a transreversal Type 2 $\bar{\rho\tau}_2(i,j,k)$, with $1 \leq i < j < k \leq {n+1}$, are operations that transform a permutation $\pi$ in the following way:
  \small{\begin{align*}
    \pi \comp \bar{\rho\tau}_1(i,j,k) = &(\ldots~\pi_{i-1}~\underline{\pi_j~\ldots~\pi_{k-1}}~\underline{{-\pi_{j-1}}~\ldots~{-\pi_{i}}}~\pi_{k}~\ldots), \\
    \pi \comp \bar{\rho\tau}_2(i,j,k) = &(\ldots~\pi_{i-1}~\underline{{-\pi_{k-1}}~\ldots~{-\pi_j}}~\underline{\pi_i~\ldots~\pi_{j-1}}~\pi_{k}~\ldots).
  \end{align*}}
\end{definition}

\begin{definition}
  Considering an unsigned permutation $\pi$, a revrev $\rho\rho(i,j,k)$, with $1 \leq i < j < k \leq {n+1}$, is an operation that when applied to $\pi$ transforms it in the permutation $\pi \comp \rho\rho(i,j,k) = (\pi_1~\ldots~\pi_{i-1}~\underline{\pi_{j-1}~\ldots~\pi_{i}}~\underline{\pi_{k-1}~\ldots~\pi_{j}}~\pi_{k}~\ldots~\pi_n)$.
\end{definition}

\begin{definition}
  Considering a signed permutation $\pi$, a revrev $\bar{\rho\rho}(i,j,k)$, with $1 \leq i < j < k \leq {n+1}$, is an operation that when applied to $\pi$ transforms it in the permutation $\pi \comp \bar{\rho\rho}(i,j,k) = (\pi_1~\ldots~\pi_{i-1}~\underline{-\pi_{j-1}~\ldots~-\pi_{i}}~\underline{-\pi_{k-1}~\ldots~-\pi_{j}}~\pi_{k}~\ldots$ $\pi_n)$.
\end{definition}

We use $\rho\tau$ or $\bar{\rho\tau}$ to denote both types of transreversals generically.


\subsection{Breakpoints}\label{sec:breakpoints}

The proofs presented in the next sections use the concept of breakpoints in a permutation. This concept is widely used in algorithms for sorting permutations by rearrangements~\cite{2009-fertin-etal}. The definition of breakpoints depends on the rearrangement model considered and whether the permutation is signed or unsigned. First, we define what an extended permutation is.

\begin{definition}
  Given a permutation $\pi$, we extend $\pi$ by adding elements $\pi_0 = 0$ and $\pi_{n+1} = n + 1$. These elements are positive when considering signed permutations. Furthermore, these elements are never affected by a rearrangement.
\end{definition}

Next, we define the two types of breakpoints.

\begin{definition}\label{def:rev_b}
An unsigned reversal breakpoint exists between a pair of consecutive elements $(\pi_i, \pi_{i+1})$ if $|\pi_{i+1} - \pi_i| \neq 1$, for $0 \leq i \leq n$.
\end{definition}

\begin{example}
  \novo{For the unsigned permutation $\pi = (0~4~3~5~1~2~6~7)$ (which is in the extended form), we have the following breakpoints from Definition~\ref{def:rev_b} (represented by the symbol $\circ$):}
  \begin{align*}
    \pi = (0 \, \circ \, 4~3  \, \circ \, 5 \, \circ \, 1~2 \, \circ \, 6~7).
  \end{align*}
\end{example}

\begin{definition}\label{def:t_b}
A transposition breakpoint (also called signed reversal breakpoint) exists between a pair of consecutive elements $(\pi_i, \pi_{i+1})$ if $\pi_{i+1} - \pi_i \neq 1$, for $0 \leq i \leq n$.
\end{definition}

\begin{example}
  \novo{For the unsigned permutation $\pi = (0~4~3~5~1~2~6~7)$ (which is in the extended form), we have the following breakpoints from Definition~\ref{def:t_b} (represented by the symbol $\circ$):}
  \begin{align*}
    \pi = (0 \, \circ \, 4 \, \circ \, 3  \, \circ \, 5 \, \circ \, 1~2 \, \circ \, 6~7).
  \end{align*}
\end{example}

\begin{example}
  \novo{For the signed permutation $\pi = (+0~{-4}~{-3}~{+5}~{+1}~{+2}~{-6}~{+7})$ (which is in the extended form), we have the following breakpoints from Definition~\ref{def:t_b} (represented by the symbol $\circ$):}
  \begin{align*}
    \pi = (+0\, \circ \,{-4}~{-3}\, \circ \,{+5}\, \circ \,{+1}~{+2}\, \circ \,{-6}\, \circ \,{+7}).
  \end{align*}
\end{example}

\novo{
We use the number of breakpoints in a permutation as an indicator of how far this permutation is from the identity permutation.
When considering signed permutations, all models use signed reversal breakpoints (Definition~\ref{def:t_b}) and, in this way, only the signed identity permutation has zero breakpoints.
Considering unsigned permutations, the reverse permutation $\eta = (n~(n-1)~\ldots~1)$ has only $2$ unsigned reversal breakpoints (Definition~\ref{def:rev_b}) and it has $n+1$ transposition breakpoints (Definition~\ref{def:t_b}). The permutation $\eta$ can be transformed into $\iota$ using only one reversal or at most two operations from the models that have transpositions alongside transreversals or revrevs, but when using only transpositions it needs a sequence of size $\Theta(n)$~\cite{2000b-meidanis-etal}. Therefore, transposition breakpoints (Definition~\ref{def:t_b}) are used for the model considering only transpositions and the other models for unsigned permutations use unsigned reversal breakpoints (Definition~\ref{def:rev_b}).
}

\begin{definition}
Given a model $\mathcal{M}$ and a permutation $\pi$, let $b_{\mathcal{M}}(\pi)$ denote the number of breakpoints in $\pi$.
\end{definition}

\begin{definition}
Given a model $\mathcal{M}$, a rearrangement $\beta \in \mathcal{M}$, and a permutation $\pi$, let $\Delta b_\mathcal{M}(\pi, \beta) = b_\mathcal{M}(\pi) - b_\mathcal{M}(\pi \comp \beta)$ denote the variation in the number of breakpoints after applying $\beta$ to $\pi$.
\end{definition}

We note that the identity permutation is the only one without breakpoints of definitions~\ref{def:rev_b} and~\ref{def:t_b}.


\begin{definition}
Considering a type of breakpoint, a \emph{strip} is a maximal sequence of elements without breakpoints between consecutive elements in the sequence.
\end{definition}

For unsigned permutations, a strip $(\pi_i~\pi_{i+1}~\ldots~\pi_j)$, with $0 \leq i < j \leq n + 1$, is called \emph{increasing} if $\pi_{k+1} > \pi_{k}$ for all $i \leq k < j$; otherwise the strip is called \emph{decreasing}. A \emph{singleton} is a strip of length one. A singleton is called increasing if it is equal to $(\pi_0)$ or $(\pi_{n+1})$, and it is called decreasing otherwise. Note that the elements $\pi_0$ and $\pi_{n+1}$ always belong to increasing strips. For signed permutations, a strip is called \emph{positive} if its elements are positive, and \emph{negative} otherwise.

\subsection{Sorting Permutations by Transpositions}

The hardness proofs presented in sections \ref{sec:trans_revrev} and \ref{sec:frag} rely on reductions from the following decision problem:

\begin{definition}
  \textit{SB3T} Problem: Given an unsigned permutation $\pi$, decide if it is possible to sort $\pi$ with a sequence of $b_{\tau}(\pi)/3$ transpositions.
\end{definition}

\cite{2012-bulteau-etal} proved that \textit{SB3T} is NP-hard by reducing from the Boolean Satisfiability (\textit{SAT}) problem.

\section{Models with Transreversal and Revrev}\label{sec:trans_revrev}

In this section, we prove that the problems with the following models are NP-hard for signed and unsigned permutations:
\begin{itemize}
  \item $\mathcal{M}_1 = \{\tau, \rho\tau \}$ or $\bar{\mathcal{M}_1} = \{\tau, \bar{\rho\tau} \}$: Transpositions and Transreversals;
  \item $\mathcal{M}_2 = \{\rho, \tau, \rho\tau \}$ or $\bar{\mathcal{M}_2} = \{\bar{\rho}, \tau, \bar{\rho\tau} \}$: Reversals, Transpositions, and Transreversals;
  \item $\mathcal{M}_3 = \{\tau, \rho\rho \}$ or $\bar{\mathcal{M}_3} = \{\tau, \bar{\rho\rho} \}$: Transpositions and Revrevs;
  \item $\mathcal{M}_4 = \{\rho, \tau, \rho\rho \}$ or $\bar{\mathcal{M}_4} = \{\bar{\rho}, \tau, \bar{\rho\rho} \}$: Reversals, Transpositions, and Revrevs;
  \item $\mathcal{M}_5 = \{\tau, \rho\tau, \rho\rho \}$ or $\bar{\mathcal{M}_5} = \{\tau, \bar{\rho\tau}, \bar{\rho\rho} \}$: Transpositions, Transreversals, and Revrevs.
  \item $\mathcal{M}_6 = \{\rho, \tau, \rho\tau, \rho\rho \}$ or $\bar{\mathcal{M}_6} = \{\bar{\rho}, \tau, \bar{\rho\tau}, \bar{\rho\rho} \}$: Reversals, Transpositions, Transreversals, and Revrevs.
\end{itemize}

Besides that, we use costs $w_1$ for all reversals and $w_2$ for all transpositions, transreversals, and revrevs. When $w_1 = w_2$, this problem is equivalent to the unweighted approach. If a model does not contain reversals, we consider that $w_1 = \infty$.

\begin{definition}
  \textit{WSR} Problem: Given a rearrangement model $\mathcal{M}$, weights $w_1$ and $w_2$, a permutation $\pi$, and a value $k$, decide if it is possible to sort $\pi$ with a sequence of rearrangements $S$, such that $w(S) \leq k$ and every rearrangement of $S$ is in $\mathcal{M}$, that is, $d_\mathcal{M}(\pi) \leq k$.
\end{definition}

Now, we present lower bounds for the sorting distance and bounds for the variation in the number of breakpoints for some permutation families.

\begin{lemma}\label{lemma:lb_signed}
  For any signed permutation $\pi$, the weights $w_1$ and $w_2$, and model $\mathcal{M} \in \{\bar{\mathcal{M}_1}, \bar{\mathcal{M}_2}, \bar{\mathcal{M}_3}, \bar{\mathcal{M}_4}, \bar{\mathcal{M}_5}, \bar{\mathcal{M}_6}\}$, we have that
  $$d_{\mathcal{M}}(\pi) \geq \min\left\{\frac{w_1}{2}, \frac{w_2}{3}\right\} b_{\mathcal{M}}(\pi).$$
\end{lemma}

\begin{proof}
  Since a reversal breaks the permutation on two positions, we have that $-2 \leq \Delta b_\mathcal{M}(\pi, \rho) \leq 2$ and $-2 \leq \Delta b_\mathcal{M}(\pi, \bar{\rho}) \leq 2$ for any reversal. If an operation $\beta$ is a transposition, transreversal, or revrev, we have that $-3 \leq \Delta b_\mathcal{M}(\pi, \beta) \leq 3$, since these operations break the permutation on three positions.

  The identity permutation is the only one without breakpoints and, consequently, a sorting sequence removes $b_\mathcal{M}(\pi)$ breakpoints. The minimum cost to remove a breakpoint is equal to $\min\left\{\frac{w_1}{2}, \frac{w_2}{3}\right\}$. Therefore, we conclude that any sorting sequence has cost greater than or equal to $\min\left\{\frac{w_1}{2}, \frac{w_2}{3}\right\} b_{\mathcal{M}}(\pi)$.
\qed\end{proof}

\begin{lemma}\label{lemma:signed_var_b}
  For any signed permutation $\pi$ such that $\pi$ has only positive strips:
  \begin{itemize}
     \item $\Delta b_{\bar{\rho}}(\pi, \bar{\rho}) \leq 0$, for any reversal $\bar{\rho}$;
     \item $\Delta b_{\bar{\rho\tau}}(\pi, \bar{\rho\tau}) \leq 1$, for any transreversal $\bar{\rho\tau}$;
     \item $\Delta b_{\bar{\rho\rho}}(\pi, \bar{\rho\rho}) \leq 1$, for any revrev $\bar{\rho\rho}$.
   \end{itemize}
\end{lemma}

\begin{proof}
  Consider a reversal $\bar{\rho}$ and let $\pi' = \pi \comp \bar{\rho} = (\pi_1~\ldots~\pi_{i-1}~\underline{{-\pi_j}~\ldots~{-\pi_i}}~\pi_{j+1}$ $\ldots~\pi_n)$, with $1 \leq i \leq j \leq n$. Suppose for the sake of contradiction that $\Delta b_{\bar{\rho}}(\pi, \bar{\rho}) > 0$, which indicates that (i) $(\pi_{i-1}, -\pi_{j})$ is not a breakpoint or (ii) $(-\pi_{i}, \pi_{j+1})$ is not a breakpoint. If $(\pi_{i-1}, -\pi_{j})$ is not a breakpoint, then $\pi_{i-1}$ and $-\pi_{j}$ must have the same sign. In the same way, if $(-\pi_{i}, \pi_{j+1})$ is not a breakpoint, then $-\pi_{i}$ and $\pi_{j+1}$ must have the same sign, which contradicts the fact that $\pi$ has only positive elements. Therefore, $\Delta b_{\bar{\rho}}(\pi, \bar{\rho}) \leq 0$.

  Consider a transreversal Type 1 $\bar{\rho\tau}_1$ and let $\pi' =$ $\pi \comp \bar{\rho\tau}_1 =$ $(\pi_1~\ldots~\pi_{i-1}$ $\underline{\pi_j~\ldots~\pi_{k-1}}$ $\underline{-\pi_{j-1}~\ldots~-\pi_{i}}$ $\pi_{k}~\ldots~\pi_n)$, with $1 \leq i < j < k \leq n+1$. Using a similar argument to the one used for reversals: the pairs $(\pi_{k-1}, -\pi_{j-1})$ and $(-\pi_{i}, \pi_k)$ must be breakpoints and only the pair $(\pi_{i-1}, \pi_j)$ may not be a breakpoint. Therefore, $\Delta b_{\bar{\rho\tau}}(\pi, \bar{\rho\tau}_1) \leq 1$. We use an analogous argument for a transreversal Type 2.

  Consider a revrev  $\bar{\rho\rho}$ and let $\pi' = \pi \comp \bar{\rho\rho} = (\pi_1~\ldots~\pi_{i-1}$ $\underline{-\pi_{j-1}~\ldots~-\pi_{i}}$ $\underline{-\pi_{k-1}~\ldots~-\pi_{j}}~\pi_{k}~\ldots~\pi_n)$, with $1 \leq i < j < k \leq n+1$. The pairs $(\pi_{i-1}, -\pi_{j-1})$ and $(-\pi_j, \pi_k)$ must be breakpoints and only the pair $(-\pi_{i}, -\pi_{k-1})$ may not be a breakpoint. Therefore, $\Delta b_{\bar{\rho\rho}}(\pi, \bar{\rho\rho}) \leq 1$.
\qed\end{proof}

\begin{theorem}\label{theorem:wsr_signed}
Considering $\mathcal{M} = \{\bar{\mathcal{M}_1}, \bar{\mathcal{M}_2}, \bar{\mathcal{M}_3}, \bar{\mathcal{M}_4}, \bar{\mathcal{M}_5},$ $ \bar{\mathcal{M}_6}\}$ and $w_2/w_1 \leq 1.5$, \textit{WSR} is NP-hard.
\end{theorem}

\begin{proof}
  \novo{
  Consider $\mathcal{M} = \bar{\mathcal{M}_6}$. The proof is similar for the other models, since our strategy is to show that if an instance is satisfied, then only transpositions are used to sort the permutation and, therefore, a similar argument can be used for the models since they have a subset of the operations allowed in $\bar{\mathcal{M}_6}$.
  }

  We now present a reduction from the \textit{SB3T} problem to \textit{WSR}. Given an instance $\pi = (\pi_1~\ldots~\pi_n)$ for \textit{SB3T}, we construct the instance $(\mathcal{M}, w_1, w_2, \pi', k)$ for \textit{WSR}, where $\pi'$ is the signed permutation $({+\pi_1}~{+\pi_2}~\ldots~{+\pi_n})$ and $k = w_2 b_{\tau}(\pi)/3$.

  Next, we show that the instance $\pi$ is sorted by $b_\tau(\pi)/3$ transpositions if and only if $d_{\mathcal{M}}(\pi') \leq w_2 b_{\tau}(\pi)/3$.

  ($\rightarrow$) If $\pi$ is sorted by a sequence $S$ of length $b_\tau(\pi)/3$, then $S$ also sorts $\pi'$, since $\pi'$ has only positive elements, and $w(S) = w_2 b_\tau(\pi)/3$ (note that $S$ only has transpositions and each transposition has cost $w_2$).

  ($\leftarrow$)
  If $\pi'$ is sorted by a sequence $S$ of cost less than or equal to $w_2 b_{\tau}(\pi)/3$, then we claim that $S$ has only transpositions and, therefore, $S$ also sorts $\pi$ and it has length $b_\tau(\pi)/3$.

  Note that the minimum cost to remove a breakpoint in $\pi'$ is equal to $\min \{ w_1/2,$ $w_2/3 \} = w_2/3$, since $w_2/w_1 \leq 1.5$. Since $\pi'$ has only positive elements, $\pi_{i+1} - \pi_i = 1$ if and only if $\pi'_{i+1} - \pi'_i = 1$, so $b_{\tau}(\pi) = b_{\mathcal{M}}(\pi')$. Therefore, the lower bound of Lemma~\ref{lemma:lb_signed} becomes $w_2 b_{\mathcal{M}}(\pi')/3 = w_2 b_{\tau}(\pi)/3$.

  We have that $w(S)$ is equal to the lower bound $w_2 b_{\tau}(\pi)/3$ and, consequently, every rearrangement of $S$ has to remove exactly $w' \times 3/w_2$ breakpoints, where $w'$ is the rearrangement cost. Note that a sorting sequence removes $b_{\tau}(\pi)/3$ breakpoints.
 Now, suppose that $S$ has an operation that is not a transposition. Let $\beta$ be the first non transposition in $S$ to be applied. Note that before $\beta$ is applied, all strips in the permutation are positive since transpositions do not change the sign of elements. By Lemma~\ref{lemma:signed_var_b}, a reversal does not remove breakpoints and a transreversal or revrev remove at most one breakpoint in permutations that have only positive strips, which contradicts the fact that every rearrangement of $S$ removes $w' \times 3/w_2$ breakpoints. Therefore, $S$ has only transpositions. Also, since $w(S) = w_2 b_{\tau}(\pi)/3$, we conclude that $S$ has length $b_{\tau}(\pi)/3$.
\qed\end{proof}

\begin{lemma}\label{lemma:lb_unsigned}
  For any unsigned permutation $\pi$, the weights $w_1$ and $w_2$, and model $\mathcal{M} \in \{{\mathcal{M}_1}, {\mathcal{M}_2}, {\mathcal{M}_3}, {\mathcal{M}_4}, {\mathcal{M}_5}, {\mathcal{M}_6}\}$, we have that
  $$d_{\mathcal{M}}(\pi) \geq \min\left\{\frac{w_1}{2}, \frac{w_2}{3}\right\} b_{\mathcal{M}}(\pi).$$
\end{lemma}

\begin{proof}
  Similar to the proof of Lemma~\ref{lemma:lb_signed}.
\qed\end{proof}

\begin{lemma}\label{lemma:unsigned_var_b}
  For any unsigned permutation $\pi$ such that $\pi$ has only increasing strips,
  \begin{itemize}
     \item $\Delta b_{{\rho}}(\pi, {\rho}) \leq 0$, for any reversal ${\rho}$;
     \item $\Delta b_{{\rho}\tau}(\pi, {\rho}\tau) \leq 1$, for any transreversal ${\rho}\tau$;
     \item $\Delta b_{{\rho}{\rho}}(\pi, {\rho}{\rho}) \leq 1$, for any revrev ${\rho}{\rho}$.
   \end{itemize}
\end{lemma}

\begin{proof}
  Consider a reversal $\rho$ and let $\pi' = \pi \comp \rho = (\pi_1~\ldots~\pi_{i-1}~\underline{{\pi_j}~\ldots~{\pi_i}}~\pi_{j+1}~\ldots$ $\pi_n)$, with $1 \leq i \leq j \leq n$. Suppose for the sake of contradiction that $\Delta b_{{\rho}}(\pi, {\rho}) > 0$, which indicates that either (i) $(\pi'_{i-1}, \pi'_{i})$ is not a breakpoint or (ii) $(\pi'_{j}, \pi'_{j+1})$ is not a breakpoint. Note that the strips in $\pi'_i, \ldots, \pi'_j$ are all decreasing, since $\pi$ has only increasing strips. Let $S = (\pi_{i'}, \ldots, \pi_{i-1})$ be the strip containing the element $\pi_{i-1}$ in $\pi$.
  If $(\pi'_{i-1}, \pi'_{i})$ is not a breakpoint, then the strip $S$ becomes equal to $S' = (\pi_{i'}, \ldots, \pi_{i-1}, \pi_{j}, \ldots, \pi_{j'})$ in $\pi'$. Since $\pi$ has only increasing strips, we have that $S$ has at least two elements or $\pi_{i-1} = \pi_0$. So, the strip $S'$ must be an increasing strip, which contradicts the fact that the strips in $\pi'_i, \ldots, \pi'_j$ are all decreasing. We reach a similar contradiction if $(\pi'_{j}, \pi'_{j+1})$ is not a breakpoint. Therefore, $\Delta b_{{\rho}}(\pi, {\rho}) \leq 0$.

  Consider a transreversal Type 1 ${\rho\tau}_1$ and let $\pi' =$ $\pi \comp {\rho\tau}_1(i,j,k) =$ $(\pi_1~\ldots~\pi_{i-1}$ $\underline{\pi_j~\ldots~\pi_{k-1}}~\underline{\pi_{j-1}~\ldots~\pi_{i}}$ $\pi_{k}~\ldots~\pi_n)$. Using a similar argument, the pairs $(\pi_{k-1},$ $\pi_{j-1})$ and $(\pi_i, \pi_k)$ are breakpoints and only the pair $(\pi_{i-1}, \pi_j)$ may not be a breakpoint. Therefore, $\Delta b_{{\rho}\tau}(\pi, {{\rho}\tau}_1) \leq 1$. We use an analogous argument for a transreversal Type 2.

  Consider a revrev  ${\rho\rho}$ and let $\pi' = \pi \comp {\rho\rho} = (\pi_1~\ldots~\pi_{i-1}$ $\underline{\pi_{j-1}~\ldots~\pi_{i}}~\underline{\pi_{k-1}~\ldots~\pi_{j}}$ $\pi_{k}~\ldots~\pi_n)$, with $1 \leq i < j < k \leq n+1$.
  The pairs $(\pi_{i-1}, \pi_{j-1})$ and $(\pi_j, \pi_k)$ must be breakpoints and only the pair $(\pi_{i}, \pi_{k-1})$ may not be a breakpoint. Therefore, $\Delta b_{{\rho\rho}}(\pi, {\rho\rho}) \leq 1$.
\qed\end{proof}

\begin{theorem}\label{theorem:wsr_unsigned}
Considering $\mathcal{M} = \{{\mathcal{M}_1}, {\mathcal{M}_2}, {\mathcal{M}_3}, {\mathcal{M}_4}, {\mathcal{M}_5},$ ${\mathcal{M}_6}\}$ and $w_2/w_1 \leq 1.5$, \textit{WSR} is NP-hard.
\end{theorem}

\begin{proof}
  Consider $\mathcal{M} = \mathcal{M}_6$. The proof is similar for the other models, since our strategy is to show that if an instance is satisfied, then only transpositions are used to sort the permutation and, therefore, a similar argument can be used for the models since they have a subset of the operations allowed in $\mathcal{M}_6$.

  We also present a reduction from the \textit{SB3T} problem to \textit{WSR}. Given an instance $\pi = (\pi_1~\ldots~\pi_n)$ for \textit{SB3T}, we construct the instance $(\mathcal{M}, w_1, w_2, \pi', k)$ for \textit{WSR}, where $k = w_2 b_{\tau}(\pi)/3$ and $\pi'$ is a permutation with $2n$ elements such that $\pi'_{2i-1} = 2\pi_{i} - 1$ and $\pi'_{2i} = 2\pi_{i}$, with $1 \leq i \leq n$.

  Next, we show that the instance $\pi$ is sorted by $b_\tau(\pi)/3$ transpositions if and only if $d_{\mathcal{M}}(\pi') \leq w_2 b_{\tau}(\pi)/3$.

  ($\rightarrow$) If $\pi$ is sorted by a sequence $S$ of length $b_\tau(\pi)/3$, then we construct the sorting sequence $S'$ such that, for every transposition $\tau(i,j,k)$ in $S$, we add the transposition $\tau(2i-1, 2j-1, 2k-1)$ in $S'$, since every element of $\pi$ was mapped into two consecutive elements of $\pi'$. In this way, $w(S') = w_2 b_{\tau}(\pi)/3$.

  ($\leftarrow$) If $\pi'$ is sorted by a sequence $S'$ of cost less than or equal to $w_2 b_{\tau}(\pi)/3$, then we claim that there exists the sequence $S$ such that $S$ has only transpositions, it sorts $\pi$, and it has length $b_\tau(\pi)/3$.

  Note that since $\pi'_{2i-1}$ and $\pi'_{2i}$ are consecutive elements, the pairs $(\pi'_{2i-1},\pi'_{2i})$ are not breakpoints, for any $1 \leq i \leq n$. Also, for $0 \leq i \leq n$:
  \begin{itemize}
    \item if $\pi_{i+1} - \pi_i = 1$, then $\pi'_{2i+1} - \pi'_{2i} = 2\pi_{i+1} - 1 - 2\pi_{i} = 1$;
    \item if $\pi_{i+1} - \pi_i > 1$, then $\pi'_{2i+1} - \pi'_{2i} = 2\pi_{i+1} - 1 - 2\pi_{i} = 2(\pi_{i+1} - \pi_i) - 1 > 1$;
    \item if $\pi_{i+1} - \pi_i < 1$, then $\pi'_{2i+1} - \pi'_{2i} = 2\pi_{i+1} - 1 - 2\pi_{i} = 2(\pi_{i+1} - \pi_i) - 1 < 1$.
  \end{itemize}

  In this way, $b_{\tau}(\pi) = b_{\mathcal{M}}(\pi')$. Note that the minimum cost to remove a breakpoint in $\pi'$ is equal to $\min \{ w_1/2, w_2/3 \} = w_2/3$, since $w_2/w_1 \leq 1.5$. So, the lower bound of Lemma~\ref{lemma:lb_unsigned} becomes $w_2 b_{\mathcal{M}}(\pi')/3 = w_2 b_{\tau}(\pi)/3$.
  We have that $w(S)$ is equal to the lower bound $w_2 b_{\tau}(\pi)/3$ and, consequently, every rearrangement of $S'$ has to remove exactly $w' \times 3/w_2$ breakpoints, where $w'$ is the rearrangement cost.

  Suppose that $S'$ has an operation that is not a transposition and let $\beta$ be the first non transposition in $S'$ to be applied. Note that before $\beta$ is applied, all strips of the permutation are increasing, since transpositions that remove $3$ breakpoints do not inverse increasing strips.
  By Lemma~\ref{lemma:unsigned_var_b}, $\beta$ does not remove $w' \times 3/w_2$ breakpoints, where $w'$ is the rearrangement cost, which is a contradiction. Therefore, $S'$ has only transpositions.

   Since $w(S') = w_2 b_{\tau}(\pi)/3$, the sequence $S'$ has $b_{\tau}(\pi)/3$ transpositions. Note that the transpositions of $S'$ do not break the pairs $(\pi'_{2i-1}, \pi'_{2i})$, for $1 \leq i \leq n$. Consider $S' = \tau'_1, \tau'_2, \ldots, \tau'_{b_{\tau}(\pi)/3}$. Now, we construct the sorting sequence $S = \tau_1, \tau_2, \ldots, \tau_{b_{\tau}(\pi)/3}$ for $\pi$, such that $\tau_x = \tau((i+1)/2, (j+1)/2, (k+1)/2)$ for $\tau'_x = \tau(i,j,k)$, with $1 \leq x \leq b_{\tau}(\pi)/3$.
\qed\end{proof}

\section{\novo{Fragmentation-Weighted Rearrangements}}\label{sec:frag}

To introduce the fragmentation cost function, we first formally define prefix, suffix, and complete rearrangements.

\begin{definition}
  A \emph{prefix reversal} is a reversal $\rho(1,j)$ or $\bar{\rho}(1,j)$, with $1 \leq j \leq n$. A \emph{suffix reversal} is a reversal $\rho(i,n)$ or $\bar{\rho}(i,n)$, with $1 \leq i \leq n$. Furthermore, a reversal $\rho(i,j)$ or $\bar{\rho}(i,j)$ is \emph{complete} when $i = 1$ and $j = n$.
\end{definition}

\begin{definition}
  A \emph{prefix transposition} is a transposition $\tau(1, j, k)$, with $1 < j < k \leq n + 1$. A \emph{suffix transposition} is a transposition $\tau(i, j, n+1)$, with $1 \leq i < j \leq n$. Furthermore, a transposition $\tau(i,j,k)$ is called \emph{complete} if $i = 1$ and $k = n + 1$.
\end{definition}

For any pair of consecutive positions $(i, i+1)$ of a permutation $\pi$, with $1 \leq i < n$, a rearrangement $\beta$ causes \emph{fragmentation} between $(i, i+1)$ if $\pi_i$ and $\pi_{i+1}$ are not adjacent in $\pi \comp \beta$.

\begin{definition}
Formally, the fragmentation cost function $f: \mathcal{M} \rightarrow \mathbb{R}$, where $\mathcal{M}$ is a rearrangement model, is defined as
\begin{align}
  f(\rho(i,j)) =
    \begin{cases}
      f^0, ~\text{if $i = 1$ and $j = n$} \\
      f^1, ~\text{if $i = 1$ and $j < n$} \\
      f^1, ~\text{if $i > 1$ and $j = n$} \\
      f^2, ~\text{if $i > 1$ and $j < n$,} \\
    \end{cases}
\end{align}
\begin{align}
  f(\tau(i,j,k)) =
    \begin{cases}
      f^1, ~\text{if $i = 1$ and $k = n+1$} \\
      f^2, ~\text{if $i = 1$ and $k < n+1$} \\
      f^2, ~\text{if $i > 1$ and $k = n+1$} \\
      f^3, ~\text{if $i > 1$ and $k < n+1$,} \\
    \end{cases}
\end{align}
where $f^0$ is a non-negative constant and $f^1, f^2$, and $f^3$ are positive constants. The superscript number in these constants indicates the number of fragmentations caused by the operation receiving that weight. We note that $f(\rho(i,j)) = f(\bar{\rho}(i,j))$, for all $i$ and $j$.
\end{definition}

\cite{2018-alexandrino-etal} considered that $f^0 = 0$, $f^1 = 1$, $f^2 = 2$, and $f^3 = 3$. In our hardness proofs, we consider the following conditions: $f^3/f^2 \leq 1.5$ and $f^3/f^1 \leq 3$. Note that the cost function considered by \cite{2018-alexandrino-etal} satisfies these conditions. In the next sections, even when omitted, we consider these conditions to be always true.

Given a rearrangement model $\mathcal{M}$, the fragmentation sort distance of a permutation $\pi$ is denoted by $d^f_\mathcal{M}(\pi)$.

\subsection{Fragmentation Breakpoints}

Now, we present the definitions of fragmentation breakpoints, which are similar to the definitions presented in Section~\ref{sec:breakpoints}, but they do not use the extended permutation.

\begin{definition}\label{def:f_rev_b}
An unsigned reversal fragmentation breakpoint exists between a pair of consecutive elements $(\pi_i, \pi_{i+1})$ if $|\pi_{i+1} - \pi_i| \neq 1$, for $1 \leq i < n$.
\end{definition}

\begin{example}
  For the unsigned permutation $\pi = (4~3~5~1~2~6)$, we have the following breakpoints from Definition~\ref{def:f_rev_b} (represented by the symbol $\circ$):
  \begin{align*}
    \pi = (4~3  \, \circ \, 5 \, \circ \, 1~2 \, \circ \, 6).
  \end{align*}
\end{example}

\begin{definition}\label{def:f_t_b}
A transposition fragmentation breakpoint (also called signed reversal fragmentation breakpoint) exists between a pair of consecutive elements $(\pi_i, \pi_{i+1})$ if $\pi_{i+1} - \pi_i \neq 1$, for $1 \leq i < n$.
\end{definition}

\begin{example}
  For the unsigned permutation $\pi = (4~3~5~1~2~6)$, we have the following breakpoints from Definition~\ref{def:f_t_b} (represented by the symbol $\circ$):
  \begin{align*}
    \pi = (4 \, \circ \, 3  \, \circ \, 5 \, \circ \, 1~2 \, \circ \, 6).
  \end{align*}
\end{example}

\begin{example}
  For the signed permutation $\pi = ({-4}~{-3}~{+5}~{+1}~{+2}~{-6})$, we have the following breakpoints from Definition~\ref{def:f_t_b} (represented by the symbol $\circ$):
  \begin{align*}
    \pi = ({-4}~{-3}\, \circ \,{+5}\, \circ \,{+1}~{+2}\, \circ \,{-6}).
  \end{align*}
\end{example}

Given a rearrangement model $\mathcal{M}$, the number of fragmentation breakpoints in a permutation $\pi$ is denoted by $b^f_{\mathcal{M}}(\pi)$ and the change in the number of fragmentation breakpoints caused by a rearrangement $\beta$ is denoted by $\Delta b^f_{\mathcal{M}}(\pi, \beta)$. The strips are defined in a similar way.

For unsigned permutations, only the identity permutation $\iota$ has no transposition fragmentation breakpoints (Definition \ref{def:f_t_b}) and only the identity permutation $\iota$ and the reverse permutation $\eta$ have no unsigned reversal fragmentation breakpoints (Definition \ref{def:f_rev_b}). Note that a complete reversal does not remove breakpoints.

For signed permutations, only the identity permutation $\iota$ and the reverse permutation $\bar{\eta}$ have no signed reversal fragmentation breakpoints (Definition \ref{def:f_t_b}). Note that, for signed permutation, we only use signed reversal fragmentation breakpoints (Definition \ref{def:f_t_b}).

\begin{lemma}[\cite{2018-alexandrino-etal}]\label{lemma:delta_b_t}
  For any permutation $\pi$ and transposition $\tau$, the maximum number of breakpoints from definitions \ref{def:f_rev_b} or \ref{def:f_t_b} removed by $\tau$ is equal to the number of fragmentations caused by this operation.
\end{lemma}

\begin{lemma}[\cite{2018-alexandrino-etal}]\label{lemma:delta_b_r}
  For any unsigned permutation $\pi$ and reversal $\rho$, the maximum number of breakpoints from definition \ref{def:f_rev_b} removed by $\rho$ is equal to the number of fragmentations caused by this operation.
\end{lemma}

\begin{lemma}[\cite{2018-alexandrino-etal}]\label{lemma:delta_b_r}
  For any signed permutation $\pi$ and reversal $\bar{\rho}$, the maximum number of breakpoints from definition \ref{def:f_t_b} removed by $\bar{\rho}$ is equal to the number of fragmentations caused by this operation.
\end{lemma}

\begin{lemma}\label{lemma:lb_fragmentation_unsigned}
  For any unsigned permutation $\pi$, we have
  \begin{align*}
    d^f_{\tau}(\pi) &\geq \frac{f^3 b^f_{\tau}(\pi)}{3},\text{~and}\\
    d^f_{\{\rho,\tau\}}(\pi) &\geq \frac{f^3 b^f_{\{\rho,\tau\}}(\pi)}{3}.
  \end{align*}
\end{lemma}

\begin{proof}
Consider the model with only transpositions. The identity permutation has no breakpoints and, therefore, a sorting sequence must remove all breakpoints from $\pi$. From Lemma~\ref{lemma:delta_b_r}, the ratio between cost and breakpoints removed is at least $min(f^3/3, f^2/2, f^1/1)$. By the restriction on the weights of the transpositions, we have that $min(f^3/3, f^2/2, f^1/1) = f^3/3$. Therefore, to remove all breakpoints, a sorting sequence has cost of at least $f^3 b^f_{\tau}(\pi) /3 $. When considering reversals, the minimum ratio between cost and breakpoints removed is also $f^3/3$ and the proof is similar.
\qed\end{proof}

\begin{lemma}\label{lemma:lb_fragmentation_signed}
  For any signed permutation $\pi$, we have
\begin{align*}
  d^f_{\{\bar{\rho},\tau\}}(\pi) \geq \frac{f^3 b^f_{\{\bar{\rho},\tau\}}(\pi)}{3}.
\end{align*}
\end{lemma}

\begin{proof}
Similar to the proof of Lemma~\ref{lemma:lb_fragmentation_unsigned}.
\qed\end{proof}

\subsection{Hardness Proofs}

Now, we prove that the problems with the following models are NP-hard considering the fragmentation cost function:

\begin{itemize}
  \item $\mathcal{M}^f_1 = \{\tau\}$: Transpositions on Unsigned Permutations;
  \item $\mathcal{M}^f_2 = \{\bar{\rho}, \tau\}$: Reversals and Transpositions on Signed Permutations;
  \item $\mathcal{M}^f_3 = \{\rho, \tau\}$: Reversals and Transpositions on Unsigned Permutations.
\end{itemize}

\begin{definition}
  \textit{FWSR} Problem: Given a rearrangement model $\mathcal{M}$, a permutation $\pi$, and a value $k$, decide if it is possible to sort $\pi$ with a sequence of rearrangements $S$, such that $f(S) \leq k$ and every rearrangement of $S$ is in $\mathcal{M}$, that is, $d^f_\mathcal{M}(\pi) \leq k$.
\end{definition}

\begin{lemma}\label{lemma:var_prefix_suffix}
  For any unsigned permutation $\pi$, such that $\pi_1 = 1$ and $\pi_n = n$, we have that $\Delta b^f_{\tau}(\pi, \tau)$ is less than the number of fragmentations caused by $\tau$, if $\tau$ is a prefix, suffix, or complete transposition.
\end{lemma}

\begin{proof}
  Consider a prefix transposition $\tau(1,j, k)$, with $1 < j < k \leq n + 1$, and let $\pi' = \pi \comp \tau(1,j, k) = (\underline{\pi_j~\ldots~\pi_{k-1}}~\underline{\pi_{1}~\ldots~\pi_{j-1}}~\pi_{k}~\ldots~\pi_n)$. Note that the number of fragmentations caused by $\tau(1,j, k)$ is equal to $2$. The pair $(\pi_{k-1}, \pi_1)$ must be a breakpoint (Definition~\ref{def:f_t_b}) in $\pi'$, since $\pi_1 = 1$ and $\pi_1 - \pi_{k-1} < 0$. Since only the pair $(\pi_{j-1}, \pi_{k})$ may not be a breakpoint, we have that $\Delta b^f_{\tau}(\pi, \tau(1, j, k)) < 2$. Note that fragmentation breakpoints do not consider the extended permutation and the pair $(\pi_0, \pi_1)$ is not considered.

  The proof is similar when considering suffix or complete transpositions.
\qed\end{proof}

Considering the model $\mathcal{M}^f_1 = \{\tau\}$, we define the problem of deciding if a permutation can be sorted with a sequence of cost equal to the the lower bound from Lemma~\ref{lemma:lb_fragmentation_unsigned}.

\begin{definition}
  \textit{FWST} Problem: Given a permutation $\pi$, decide if it is possible to sort $\pi$ with a sequence of transpositions $S$, such that $f(S) = f^3 b^f_{\tau}(\pi)/{3}$, that is, $d^f_{\mathcal{M}^f_1}(\pi) = f^3 b^f_{\tau}(\pi)/{3}$.
\end{definition}

\begin{theorem}\label{theorem:fwsr_lbt}
  The \textit{FWST} problem is NP-hard.
\end{theorem}

\begin{proof}
  Given an instance $\pi$ for \textit{SB3T}, we construct the instance $\pi'$ for \textit{FWST}, where $\pi' = (1~(\pi_1+1)~(\pi_2+1)~\ldots~(\pi_n+1)~{n+2})$.

  Note that, for $1 \leq i \leq n+2$, $\pi'_i = \pi_{i-1} + 1$ and, so, $\pi'_{i+1} - \pi'_{i} = (\pi_i + 1) - (\pi_{i-1} + 1) = \pi_i - \pi_{i-1}$. Therefore, the pair $(\pi_i, \pi_{i+1})$ is a breakpoint (Definition~\ref{def:t_b}) if and only if $(\pi'_{i+1}, \pi'_{i+2})$ is a breakpoint (Definition~\ref{def:f_t_b}), for $0 \leq i \leq n$, and $b_{\tau}(\pi) = b^f_{\tau}(\pi')$. Observe that the Definition~\ref{def:f_t_b} does not consider the extended permutation.

  Next, we show that the instance $\pi$ is sorted by $b_{\tau}(\pi)/3$ transpositions if and only if $d^f_{\tau}(\pi') = f^3 b^f_{\tau}(\pi')/{3}$.

  ($\rightarrow$) If $\pi$ is sorted by a sequence $S$ of length $b_\tau(\pi)/3$, then we construct a sorting sequence $S'$ for $\pi'$ by mapping each transposition $\tau(i,j,k)$ in $S$ into the transposition $\tau(i+1, j+1, k+1)$ in $S'$. Note that each transposition in $S'$ has fragmentation cost of $f^3$ and so $f(S') = f^3 b^f_{\tau}(\pi')/{3}$.

  ($\leftarrow$) If $\pi'$ is sorted by a sequence $S'$ such that $f(S') =  f^3 b^f_{\tau}(\pi')/{3}$, then we claim that all transpositions of $S'$ have cost $f^3$ and, consequently, there exists sequence $S$ of length $b_\tau(\pi)/3$ which sorts $\pi$.

  Suppose for the sake of contradiction that $S'$ has a prefix, suffix, or complete transposition. Consider $S' = \tau'_1, \tau'_2, \ldots, \tau'_m$. Since $f(S')$ is exactly the lower bound of Lemma~\ref{lemma:lb_fragmentation_unsigned}, we have that for every transposition $\tau'$ in $S'$, the number of breakpoints removed by $\tau'$ must be equal to the number of fragmentations caused by $\tau'$. Let $\tau'_i$ be the first prefix, suffix, or complete transposition of $S'$ applied to the permutation. Before $\tau'_i$ is applied, we have that the first and the last elements of the permutation are in the correct position and, according to Lemma~\ref{lemma:var_prefix_suffix}, the variation in the number of breakpoints caused by $\tau'$ is less than the number of fragmentations caused by it, which is a contradiction. Therefore, all transpositions of $S'$ have cost $f^3$ and $|S'| = b_\tau(\pi)/3$.

  Now, we construct the sorting sequence $S = \tau_1, \tau_2, \ldots, \tau_{b_\tau(\pi)/3}$ for $\pi$, such that $\tau_x = \tau(i - 1, j - 1, k - 1)$ for $\tau'_x(i, j, k)$, with $1 \leq x \leq b_{\tau}(\pi)/3$.
\qed\end{proof}

\begin{corollary}
  Considering the model $\mathcal{M}^f_1 = \{\tau\}$, the \textit{FWSR} problem is NP-hard.
\end{corollary}

Now, we present hardness proofs for the \textit{FWSR} problem, considering the models $\mathcal{M}^f_2 = \{\bar{\rho}, \tau\}$ and $\mathcal{M}^f_3 = \{\rho, \tau\}$, using a reduction from the \textit{FWST} problem.

\begin{lemma}[\cite{2018-alexandrino-etal}]\label{lemma:one_complete_reversal}
  For any sequence $S$ with more than one complete reversal, there exists a sequence $S'$ such that $S$ and $S'$ have the same effect and $S'$ has at most one complete reversal, which is the last rearrangement of $S'$, if it exists.
\end{lemma}

\begin{lemma}\label{lemma:f_pos_strips}
  For any signed permutation $\pi$, such that all strips of $\pi$ are positive, we have that $\Delta b^f_{\mathcal{M}^f_2}(\pi, \bar{\rho})$ is less than the number of fragmentations caused by $\bar{\rho}$, if $\bar{\rho}$ is a prefix or suffix reversal.
\end{lemma}

\begin{proof}
  Using a similar argument to the one used in Lemma~\ref{lemma:signed_var_b}, we know that a reversal does not remove breakpoints of $\pi$ and the results follows.
  %
\qed\end{proof}

\begin{theorem}\label{theorem:fwsr_rt}
  Considering the model $\mathcal{M}^f_2 = \{\bar{\rho}, \tau\}$, the \textit{FWSR} problem is NP-hard.
\end{theorem}

\begin{proof}
  Given an instance $\pi$ for the \textit{FWST} problem, we construct the instance $(\mathcal{M}^f_2, \pi', k)$ for \textit{FWSR}, where $\pi' = (+\pi_1~\ldots~+\pi_{n})$ and $k = f^3 b^f_{\tau}(\pi) / 3$.

  Note that $b^f_{\tau}(\pi) = b^f_{\mathcal{M}^f_2}(\pi')$, since the definition of breakpoint is the same and all elements of $\pi'$ are positive.

  Next, we show that $d^f_{\tau}(\pi) = f^3 b^f_{\tau}(\pi) / 3$ if, and only if, $d^f_{\mathcal{M}^f_2}(\pi) = f^3 b^f_{\tau}(\pi) / 3$.

  $(\rightarrow)$ If $\pi$ is sorted by a sequence of transpositions $S$ of cost $f^3 b^f_{\tau}(\pi) / 3$, then $S$ also sorts $\pi'$ with the same cost.

  $(\leftarrow)$ If $\pi'$ is sorted by a sequence $S$ of cost $f^3 b^f_{\tau}(\pi) / 3$, then we claim that $S$ has only transpositions and also sorts $\pi$ with cost $f^3 b^f_{\tau}(\pi) / 3$.

  Consider, without loss of generality, that $S$ has at most one complete reversal, which is the last rearrangement of $S$, if it exists (Lemma~\ref{lemma:one_complete_reversal}).

  Since $f(S)$ is exactly the lower bound of Lemma~\ref{lemma:lb_fragmentation_unsigned}, we have that for every operation $\beta$ in $S$, the number of breakpoints removed by $\beta$ must be equal to the number of fragmentations caused by it. Suppose for the sake of contradiction that $S$ has a reversal that is not a complete reversal. Let $\bar{\rho}$ be the first reversal of $S$ to be applied. Before $\bar{\rho}$ is applied, the permutation has only positive strips and, according to Lemma~\ref{lemma:f_pos_strips}, this rearrangement does not remove breakpoints, which is a contradiction. Therefore, $S$ has only transpositions, except for the last rearrangement that may be a complete reversal. Suppose that the last rearrangement $\beta'$ in $S$ is a complete reversal. Since the other rearrangements of $S$ are transpositions, all elements of the permutation are positive before applying $\beta'$, and a complete reversal would turn every element into a negative element, which contradicts the fact that $S$ is a sorting sequence for $\pi'$. Therefore, $S$ has only transpositions and also sorts $\pi$ with cost $f^3 b^f_{\tau}(\pi) / 3$.
\qed\end{proof}

\begin{lemma}\label{lemma:f_inc_strips}
  For any unsigned permutation $\pi$, such that all strips of $\pi$ are increasing, we have that $\Delta b^f_{\mathcal{M}^f_3}(\pi, \rho)$ is less than the number of fragmentations caused by $\rho$, if $\rho$ is a prefix or suffix reversal.
\end{lemma}

\begin{proof}
  Using a similar argument to the one used in Lemma~\ref{lemma:unsigned_var_b}, we know that a reversal does not remove breakpoints of $\pi$ and the results follows.
\qed\end{proof}

Note that, for any unsigned permutation $\pi$ in which all strips of $\pi$ are increasing, applying a transposition $\tau$ in $\pi$, such that $\Delta b^f_{\mathcal{M}^f_3}(\pi, \tau)$ is maximum, does not turn any increasing strip into a decreasing strip. However, a complete reversal applied to $\pi$ turns all increasing strips into decreasing strips.

\begin{theorem}
  Considering the model $\mathcal{M}^f_3 = \{\rho, \tau\}$, the \textit{FWSR} problem is NP-hard.
\end{theorem}

\begin{proof}
Similar to the proof of theorems~\ref{theorem:wsr_unsigned} and~\ref{theorem:fwsr_rt}, using Lemma~\ref{lemma:f_inc_strips}.
\qed\end{proof}

\section{Conclusions}\label{sec:conclusion}

We showed that the problems of Sorting (Signed or Unsigned) Permutations by Rearrangements are NP-hard for twelve rearrangement models, which include transpositions alongside reversals, transreversals, and revrevs, considering that a reversal has cost $w_1$, the other rearrangements have cost $w_2$ and $w_2/w_1 \leq 1.5$. Moreover, we presented hardness proofs for the problems of Sorting Permutations by Fragmentation-Weighted Rearrangements, considering transpositions and the combination of reversals and transpositions, for some combinations of weights.

The complexity of the first set of problems remains open when $w_2/w_1 > 1.5$, and the complexity of the fragmentation-weighted problems remains open for the model with only reversals. Another future work direction is to study the hardness of approximation for the optimization problems of genome rearrangements, including models with only transpositions and transpositions alongside reversals.

\section*{Acknowledgments}

This work was supported by the National Council of Technological and Scientific Development, CNPq (grants
 400487/2016-0 and 
 425340/2016-3
),
the Coordena\c{c}\~{a}o de Aperfei\c{c}oamento de Pessoal de N\'{i}vel Superior - Brasil (CAPES) - Finance Code 001
, and the S\~ao Paulo Research Foundation, FAPESP (grants %
2013/08293-7
, 2015/11937-9
, 2017/12646-3
, and 2019/27331-3
).

\end{document}